\documentclass{article}
\title{Fast quantum algorithms for approximating some irreducible
  representations of groups}
\author{Stephen P. Jordan\footnote{Parts of this work were completed
    at MIT's Center for Theoretical Physics and RIKEN's Digital
    Materials Laboratory.}\\
{\em \small Institute for Quantum Information, California Institute
 of Technology.} \small \texttt{sjordan@caltech.edu}}

\date{}

\usepackage{graphicx, amssymb, amsmath, bm, fullpage, dsfont}

\input{Qcircuit}

\newcommand{\captionfonts}{\small}

\makeatletter  
\long\def\@makecaption#1#2{%
  \vskip\abovecaptionskip
  \sbox\@tempboxa{{\captionfonts #1: #2}}%
  \ifdim \wd\@tempboxa >\hsize
    {\captionfonts #1: #2\par}
  \else
    \hbox to\hsize{\hfil\box\@tempboxa\hfil}%
  \fi
  \vskip\belowcaptionskip}
\makeatother   

\newenvironment{proof}{\noindent \textbf{Proof:}}{$\Box$}

\begin{document}
\bibliographystyle{plain}
\maketitle
\newcommand{\ud}{\mathrm{d}}
\newcommand{\braket}[2]{\langle #1|#2\rangle}
\newcommand{\Bra}[1]{\left<#1\right|}
\newcommand{\Ket}[1]{\left|#1\right>}
\newcommand{\Braket}[2]{\left< #1 \right| #2 \right>}
\renewcommand{\th}{^\mathrm{th}}
\newcommand{\tr}{\mathrm{Tr}}
\newcommand{\id}{\mathds{1}}

\newtheorem{theorem}{Theorem}
\newtheorem{proposition}{Proposition}
\newtheorem{hypothesis}{Hypothesis}

\begin{abstract}
We consider the quantum complexity of estimating matrix
elements of unitary irreducible representations of groups. For several
finite groups including the symmetric group, quantum Fourier
transforms yield efficient solutions to this problem. Furthermore,
quantum Schur transforms yield efficient solutions for certain irreducible
representations of the unitary group. Beyond this, we obtain
$\mathrm{poly}(n)$-time quantum algorithms for approximating matrix
elements from all the irreducible representations of the alternating
group $A_n$, and all the irreducible representations of polynomial
highest weight of $U(n)$, $SU(n)$, and $SO(n)$. These quantum
algorithms offer exponential speedup in worst case complexity over the
fastest known classical algorithms. On the other hand, we show that
average case instances are classically easy, and that the techniques
analyzed here do not offer a speedup over classical computation for
the estimation of group characters.
\end{abstract}

\section{Introduction}

Explicit representations of groups have many uses in physics,
chemistry, and mathematics. All representations of finite groups and
compact linear groups can be expressed as unitary matrices given an
appropriate choice of basis\cite{Artin9}. This makes them natural
candidates for implementation using quantum circuits. Here we show
that polynomial size quantum circuits can implement:
\begin{itemize}
\item The irreducible representations of any finite group which has an
  efficient quantum Fourier transform. This includes the symmetric
  group $S_n$.
\item The irreducible representations of the alternating group $A_n$.
\item The irreducible representations of polynomial highest weight of
  the unitary $U(n)$, special unitary $SU(n)$, and special
  orthogonal $SO(n)$ groups.
\end{itemize}
Using these quantum circuits one can find a polynomially precise
additive approximation to any matrix element of these representations
by repeating a simple measurement called the Hadamard test, as
described in section \ref{hadamard}.

More precisely, for the finite groups $S_n$ and $A_n$ we obtain
any matrix element of any irreducible representation to within $\pm
\epsilon$ in time that scales polynomially in $1/\epsilon$ and
$n$. For the Lie groups $U(n)$, $SU(n)$, and $SO(n)$ we obtain any 
matrix element of any irreducible representation of polynomial highest
weight to within $\pm \epsilon$ in time that scales polynomially in
$1/\epsilon$ and $n$. Because the representations considered are of
exponentially large dimension, one cannot efficiently find these
matrix elements by classically multiplying the matrices representing a
set of generators. Note that, many computer science applications use
multiplicative approximations. In this case, one computes an estimate
$\tilde{x}$ of a quantity $x$ with the requirement that $(1-\epsilon)x
\leq \tilde{x} \leq (1-\epsilon)x$. The approximations obtained in
this paper are all additive rather than multiplicative. For some
problems, the computational complexity of additive approximations can
differ greatly from that of mulitplicative
approximations\cite{Bordewich, Tutte}.

For exponentially large unitary matrices, the typical matrix element
is exponentially small. Thus for average instances, a
polynomially precise additive approximation provides almost no
information. However, it is common that the worst case instances
of a problem are hard whereas the average case instances are
trivial. In section \ref{complexity} I narrow down a class of
potentially hard instances for the problem of additively approximating
the matrix elements of the irreducible representations of the symmetric group to
polynomial precision. I also present a classical randomized algorithm
to estimate normalized characters of the symmetric group
$S_n$ to within $\pm \epsilon$ in $\mathrm{poly}(n,1/\epsilon)$
time. (The character is normalized by dividing by the dimension of the
representation, so that the character of the identity element of the
group is 1.) Thus, the techniques described here for evaluating matrix
elements of irreducible representations of groups on quantum computers
do not provide an obvious quantum speedup for the evaluation of the
characters of $S_n$.

Our results on the symmetric group relate closely to the quantum
complexity of evaluating Jones polynomials and other topological
invariants. Certain problems of approximating Jones and HOMFLY
polynomials can be reduced to the approximation of matrix
elements or characters of the Jones-Wenzl representation of the braid
group, which is a $q$-deformation of certain irreducible
representations of the symmetric group
\cite{Aharonov1,Yard,Shor_Jordan,Jordan_Wocjan}. Figure
\ref{complexities} compares the complexity of estimating matrix
elements and characters of the Jones-Wenzl representation of the braid
group to the complexity of the corresponding problems for the
symmetric group. Exact complexity characterizations (\emph{i.e.}
completeness results) are not known for all of these problems, and the
exact relationships between the complexity classes referenced in figure
\ref{complexities} are not rigorously known. Nevertheless, the results
seem to suggest that in general the matrix elements are harder to
approximate than the normalized characters, and that the Jones-Wenzl
representation of braid group is computationally harder than the
corresponding irreducible representations of the symmetric group.

\begin{figure}
\begin{center}
\begin{tabular}{c|c|c}
 & symmetric & braid \\
\hline
matrix elements & in BQP & BQP-complete \cite{Aharonov1, Yard} \\
\hline
normalized characters & in BPP & DQC1-complete \cite{Shor_Jordan, Jordan_Wocjan}
\end{tabular}
\end{center}
\caption{\label{complexities} The complexity results on the symmetric
  group refer arbitrary irreducible representations in Young's orthogonal
  form. The results on the braid group
  refer to the Jones-Wenzl representations, which give rise to Jones
  and HOMFLY polynomials. The complexity class DQC1 is the set of
  problems solvable in polynomial time on a one clean qubit
  computer. It is generally believed that one clean qubit computers
  are weaker than standard quantum computers but still capable of
  solving some problems outside of BPP.} 
\end{figure}

\section{Hadamard Test}
\label{hadamard}

The Hadamard test is a standard technique in quantum computation for
approximating matrix elements of unitary transformations. Suppose we
have an efficient quantum circuit implementing a unitary
transformation $U$, and an efficient procedure for preparing the state
$\ket{\psi}$. We can then approximate the real part of $\bra{\psi} U
\ket{\psi}$ using the following quantum circuit.
\[
\mbox{\Qcircuit @C=1em @R=.7em {
  \lstick{\frac{1}{\sqrt{2}}(\ket{0}+\ket{1})} & \qw & \ctrl{1} & \gate{H} & \meter \\
  \lstick{\ket{\psi}} & {/} \qw  & \gate{U} & {/} \qw & \qw & }}
\]
The probability of measuring $\ket{0}$ is
\[
p_0 = \frac{1+\mathrm{Re}(\bra{\psi}U\ket{\psi})}{2}.
\]
Thus, one can obtain the real part of
$\bra{\psi}U\ket{\psi}$ to precision $\epsilon$ by
making $O(1/\epsilon^2)$ measurements and counting what fraction of
the measurement outcomes are $\ket{0}$. Similarly, if the control bit
is instead initialized to $\frac{1}{\sqrt{2}} (\ket{0} - i \ket{1})$,
one can estimate the imaginary part of $\bra{\psi}U\ket{\psi}$. Thus
the problem of estimating matrix elements of unitary representations
of groups reduces to the problem of implementing these representations
with efficient quantum circuits.

\section{Fourier Transforms}
\label{fourier}

Let $G$ be a finite group and let $\hat{G}$ be the set of all
irreducible representations of $G$. We choose a basis for the
representations such that for any $\rho \in \hat{G}$ and $g \in G$,
$g$ is represented by a $d_\rho \times d_\rho$ unitary matrix with
entries $\rho_{i,j}(g)$. The quantum Fourier transform over $G$ is the
following unitary operator\cite{Nielsen_Chuang}
\[
U_{\mathrm{FT}} = \sum_{g \in G} \sum_{\rho \in \hat{G}} 
\sum_{i,j = 1}^{d_\rho} \sqrt{\frac{d_\rho}{|G|}}
\rho_{i,j}(g) \ket{\rho,i,j}\bra{g}.
\]
Here $\ket{g}$ is a computational basis state (bitstring) indexing the element $g$
of $G$. Similarly, $\ket{\rho,i,j}$ is three bitstrings, one indexing
the element $\rho \in \hat{G}$, and two writing out the numbers $i$
and $j$ in binary. The standard discrete Fourier transform is the
special case where $G$ is a cyclic group.

The regular representation of any $g \in G$ is
\[
U_g = \sum_{h \in G} \ket{gh}\bra{h}.
\]
A short calculation shows
\[
U_{\mathrm{FT}} U_g U_{\mathrm{FT}}^{-1} = \sum_{\rho \in \hat{G}}
\sum_{i,j = 1}^{d_\rho} \sum_{i',j'=1}^{d_\rho} \delta_{j,j'}
\rho_{i,i'}(g^{-1}) \ket{\rho,i,j}\bra{\rho,i',j'}.
\]
In other words, by conjugating the regular representation of $g$ with
the quantum Fourier transform, one recovers the direct sum of all
irreducible representations of $g^{-1}$.

Given an efficient quantum circuit implementing $U_{\mathrm{FT}}$ one
can thus efficiently estimate any matrix element of any irreducible
representation of $G$ using the Hadamard test. Quantum circuits
implementing the Fourier transform in $\mathrm{polylog}(|G|)$ time are
known for the symmetric group\cite{Beals} and several other
groups\cite{generalft}. The matrix elements of the representations
depend on a choice of basis. The bases used in quantum Fourier
transforms are subgroup adapted (see \cite{generalft}). In particular, the
symmetric group Fourier transform described in \cite{Beals} uses the
Young-Yamanouchi basis, also known as Young's orthogonal form.

In section \ref{althyp} we describe a more direct quantum circuit
implementation of the irreducible representations of the symmetric
group, which generalizes to yield efficient implementations for the
alternating group.

\section{Schur Transform}
\label{Schurxform}

Let $\mathcal{H}$ be the Hilbert space of $n$ $d$-dimensional qudits.
\[
\mathcal{H} = (\mathbb{C}^d)^{\otimes n}.
\]
We can act on this Hilbert space by choosing an element $u \in U(d)$
and applying it to each qudit.
\[
\ket{\psi} \to u^{\otimes n} \ket{\psi}
\]
We can also act on this Hilbert space by choosing an
element $\pi \in S_n$ and correspondingly permuting the $n$
qudits. 
\[
\ket{\psi} \to M_\pi \ket{\psi}
\]
$u^{\otimes n}$ and $M_\pi$ are reducible unitary $nd$-dimensional
representations of $U(d)$ and $S_n$, respectively. These two actions
on $\mathcal{H}$ commute. 

The irreducible representations of $S_n$ are in bijective correspondence
with the partitions of $n$. Any partition of $n$ into $d$ parts
indexes a unique irreducible representation of $U(d)$. $U(d)$ has
infinitely many irreducible representations, so these partitions only
index a special subset of them. As discussed in \cite{Schur}, there
exists a unitary change of basis $U_{\mathrm{Schur}}$ such that
\[
U_{\mathrm{Schur}} M_\pi u^{\otimes n} U_{\mathrm{Schur}}^{-1} =
\bigoplus_\lambda \rho_\lambda(\pi) \otimes \nu_\lambda(u),
\]
where $\lambda$ ranges over all partitions of $n$ into $d$
parts.

As shown in \cite{Schur}, $U_{\mathrm{Schur}}$ can be implemented by a
$\mathrm{poly}(n,d)$ size quantum circuit. Thus, using the Hadamard
test, one can efficiently obtain matrix elements of
these representations of the symmetric and unitary groups.

\section{Complexity of Symmetric Group Representations}
\label{complexity}

As described in section \ref{fourier}, quantum computers can solve the
following problem with probability $1-\delta$ in
$\mathrm{poly}(n,1/\epsilon,\log(1/\delta))$ time. Note that standard
Young tableaux index the Young-Yamanouchi basis vectors, as discussed
in section \ref{Young}.\\ \\ 
\noindent
\begin{minipage}[c]{\textwidth}
\textbf{Problem 1:} Approximate a matrix element in the Young-Yamanouchi
basis of an irreducible representation for the symmetric group $S_n$. \\
\textbf{Input:} A Young diagram specifying the irreducible
representation, a permutation from $S_n$, a pair of standard Young
tableaux indicating the desired matrix element, and a polynomially small
parameter $\epsilon$.\\
\textbf{Output:} The specified matrix element to within $\pm
\epsilon$. \\
\end{minipage}
It appears that no polynomial time classical algorithm for this
problem is known. Due mainly to applications in quantum chemistry,
many exponential time classical algorithms for the exact computation
of entire matrices from representations of the symmetric group have
been developed\cite{Hamermesh, Boerner, Wu_Zhang1, Wu_Zhang2, Egecioglu,
  Clifton, Rettrup, Pauncz}. There appears to be no literature on the
computation or approximation of individual matrix elements of
representations of $S_n$. 

On the other hand, the precision of approximation achieved by the
quantum algorithm is trivial for average instances. We can see this as
follows. Let $\lambda$ be a Young diagram of $n$ boxes, let
$\rho_\lambda$ be the corresponding irreducible representation of
$S_n$, and let $d_\lambda$ be the dimension of $\rho_\lambda$. For any
$\pi \in S_n$, the root mean square of the matrix elements of
$\rho_\lambda(\pi)$ is
\[
\mathrm{RMS}(\rho_\lambda(\pi)) = \sqrt{\frac{1}{d_\lambda^2}
  \sum_{a,b \in B} | \bra{a} \rho_\lambda(\pi) \ket{b} |^2}, 
\]
where $B$ is any complete orthonormal basis for the vector space on
which $\rho_\lambda$ acts. We see that
\[
\sum_{a \in B} | \bra{a} \rho_\lambda(\pi) \ket{b} |^2 = 1
\]
since, by the unitarity of $\rho_\lambda(\pi)$, this is just the norm
of $\ket{b}$. Thus, 
\begin{equation}
\label{small_rms}
\mathrm{RMS}(\rho_\lambda(\pi)) = \sqrt{\frac{1}{d_\lambda^2}
  \sum_{b \in B} 1} = \frac{1}{\sqrt{d_\lambda}}.
\end{equation}
The interesting instances of problem 1 are those in which $d_\lambda$
is exponentially large. In these instances, the typical
matrix element is exponentially small, by equation
\ref{small_rms}. Running the quantum algorithm yields polynomial
precision, thus one could instead simply guess zero every time, with
similar results.

That the average case instances are trivial does not mean that the
algorithm is trivial. Hard problems that are trivial on average are a
common occurrence. The most relevant example of this is the problem
of estimating a knot invariant called the Jones polynomial. A certain
problem of estimating the Jones polynomial of knots is
BQP-complete\cite{Freedman, Aharonov1, Aharonov2}. The Jones
polynomial algorithm is based on estimating matrix elements of certain
representations of the braid group to polynomial precision. On average
these matrix elements are exponentially small. Nevertheless, the
BQP-hardness of the Jones polynomial problem shows that the worst-case
instances are as hard as any problem in BQP.

By analogy to the results on Jones polynomials, one might ask ask
whether problem 1 is BQP-hard. The existing proofs of BQP-hardness of
Jones polynomial estimation rely on the fact that the relevant
representations of the braid group are dense in the corresponding unitary
group. Thus, one can construct a braid whose representation implements
approximately the same unitary as any given quantum
circuit. Furthermore, it turns out that the number of crossings needed
to achieve a good approximation scales only polynomially with the
number of quantum gates in the circuit. Unlike the braid group, the
symmetric group is finite. Thus, no representation of it can be dense
in a continuous group. Hence, if the problem of estimating matrix
elements of the symmetric group is BQP-hard, the proof will have to
proceed along very different lines than the BQP-hardness proof for
Jones polynomials.

Lacking a hardness proof, the next best thing is to identify a class
of instances in which the matrix elements are large enough to make the
approximation nontrivial. As shown below, we can do this using the
asymptotic character theory of the symmetric group. Note that we need
not worry about the matrix elements being too large, because even if
we know \emph{a priori} that a given matrix element has magnitude 1,
it could still be nontrivial to compute its sign.

Let $\pi$ be a permutation in $S_n$, and let $\lambda$ be a Young
diagram of $n$ boxes. The character
\[
\chi_\lambda(\pi) = \tr(\rho_\lambda(\pi))
\]
is clearly independent of the basis in which $\rho_\lambda$ is
expressed. Furthermore, the character of a group element depends only
on the conjugacy class of the group element, because for any
representation $\rho$,
\[
\tr(\rho(h g h^{-1})) = \tr(\rho(h) \rho(g) \rho(h)^{-1}) = \tr(\rho(g)).
\]

\begin{figure}
\begin{center}
\includegraphics[width=0.32\textwidth]{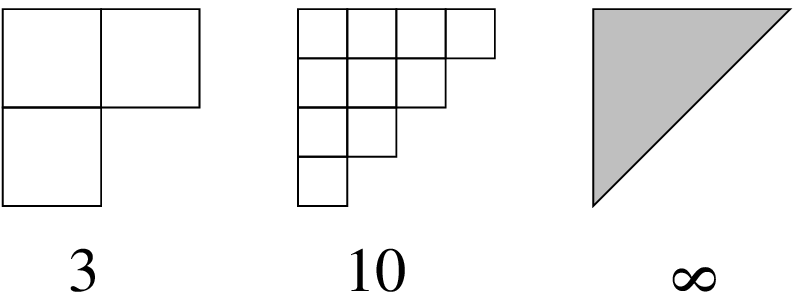}
\caption{\label{shapeconverge} Here is a sequence of Young diagrams,
 such that as the number of boxes increases, the Young diagram
 converges asymptotically to some fixed shape, in this case a
 triangle.}
\end{center}
\end{figure}

To understand the behavior of the characters of $S_n$ as $n$ becomes
large, consider a sequence of Young diagrams
$\lambda_1,\lambda_2,\lambda_3,\ldots$, where $\lambda_n$ has $n$
boxes. Suppose that the diagram $\lambda_n$, when scaled down by
a factor of $1/\sqrt{n}$, converges to a fixed shape $\omega$ in the
limit of large $n$, as illustrated in figure \ref{shapeconverge}. Let
$d_{\lambda_n}$ be the dimension of the irreducible representation
corresponding to Young diagram $\lambda_n$. Let $\pi$ be a permutation
in $S_k$. We can also consider $\pi$ to be an element of $S_n$ for any
$n > k$ which leaves the remaining $n-k$ objects fixed. As shown by
Biane\cite{Biane},
\begin{equation}
\label{Bianeformula}
\frac{\chi_{\lambda_n}(\pi)}{d_{\lambda_n}} =
C_\pi(\omega) n^{-|\pi|/2} + O(n^{-|\pi|/2-1}).
\end{equation}
Here $|\pi|$ denotes the minimum number of transpositions needed to
obtain $\pi$. Note that these are general transpositions, not
transpositions of neighbors. $C_{\pi}(\omega)$ is a constant that only
depends on $\pi \in S_k$ and the shape $\omega$. A precise definition
of what it means for the sequence to converge to a fixed shape is
given in \cite{Biane}, but for present purposes, the intuitive picture
of figure \ref{shapeconverge} should be sufficient.

$\chi_{\lambda_n}(\pi)/d_{\lambda_n}$ is the average of the matrix
elements on the diagonal of $\rho_{\lambda_n}(\pi)$. In the present
setting, where $\pi$ is fixed, $\chi_{\lambda_n}(\pi)/d_{\lambda_n}$ 
shrinks only polynomially with $n$. Thus polynomial precision is
sufficient to provide nontrivial estimates of these matrix
elements. Nevertheless, finding diagonal matrix elements of
$\rho_{\lambda_n}(\pi)$ for fixed $\pi$ and large $n$ is not
computationally hard. This is because, as discussed in section
\ref{althyp}, the Young-Yamanouchi basis is subgroup adapted. Thus, for
any $\pi$ which leaves all bit the first $k$ objects fixed,
$\rho_{\lambda_n}(\pi)$ is a direct sum of irreducible representations
of $\pi$ in $S_k$. Because $k$ is fixed, any irreducible
representations of $S_k$ has dimension $O(1)$ and can therefore be
computed in $O(1)$ time by multiplying the matrices representing
transpositions.

To produce a candidate class of hard instances of problem 1, we recall
that the character $\chi_{\lambda_n}(\pi)$ depends only on the
conjugacy class of $\pi$. Thus, we consider $\pi'$ conjugate to
$\pi$. Like $\pi \in S_n$, $\pi' \in S_n$ leaves at least $n-k$
objects fixed, and the representations $\chi_{\lambda_n}(\pi')$ have
diagonal matrix elements with polynomially small average
value. However, the objects left fixed by $\pi'$ need not be
$k+1,k+2,\ldots,n$. Indeed, $\pi'$ can be chosen so that the object 
$n$ is not left fixed, in which case $\rho_{\lambda_n}(\pi')$ cannot
be written as the direct sum of irreducible representations of $S_m$
for any $m < n$. 

There is an additional simple way in which an instance of problem 1
can fail to be hard. Let $r(\pi)$ be the minimal number of transpositions of
neighbors needed to construct the permutation $\pi$. If $r(\pi)$ is
constant or logarithmic, then the matrix elements of the irreducible
representations of $\pi$ can be computed classically in polynomial
time by direct recursive application of equation \ref{rule}. For a
class of hard instances of problem 1 I propose the following.\\ \\
\begin{minipage}[c]{\textwidth}
\begin{hypothesis}
Let $\pi$ be a permutation in $S_n$. We consider it to permute a
series of objects numbered $1,2,3,\ldots,n$. Let $s(\pi)$ be the
number of objects that $\pi$ does not leave fixed. Let $l(\pi)$ be the
largest numbered object that $\pi$ does not leave fixed. Let $r(\pi)$
be the minimum number of transpositions of neighbors needed to
construct $\pi$. Let $\lambda$ be a Young diagram of $n$ boxes, and
let $\rho_\lambda$ be the corresponding $d_\lambda$-dimensional
irreducible representation of $S_n$. I propose the problems of
estimating the diagonal matrix elements of $\rho_\lambda(\pi)$ such
that $s(\pi) = O(1)$, $l(\pi) = \Omega(n)$, and $r(\pi) = \Omega(n)$
as a possible class of instances of problem 1 not solvable classically
in polynomial time.
\end{hypothesis}
\vspace{6pt}
\end{minipage}
Although this hypothesis contains many restrictions on $\pi$, it is
clear that permutations satisfying all of these conditions exist. One
simple example is the permutation that transposes 1 with $n$.

\section{Characters of the Symmetric Group}
\label{characters}

Because characters do not depend on a choice of basis, the
computational complexity of estimating characters is especially
interesting. Hepler\cite{Hepler} showed that computing the characters
of the symmetric group exactly is \#P-hard. It is clear that an
algorithm for efficiently approximating matrix elements of a
representation can aid in approximating the corresponding
character. Specifically, the quantum algorithm for problem 1 yields an
efficient solution for the following problem.
\\ \\
\noindent
\begin{minipage}[c]{\textwidth}
\textbf{Problem 2:} Approximate a character for the symmetric group $S_n$.\\
\textbf{Input:} A Young diagram $\lambda$ specifying the irreducible
representation, a permutation $\pi$ from $S_n$, and a polynomially
small parameter $\epsilon$.\\
\textbf{Output:} Let $\chi^\lambda(\pi)$ be the character, and let
$d_\lambda$ be the dimension of the irreducible representation. The
output $\chi_{\mathrm{out}}$ must satisfy $|\chi_{\mathrm{out}} -
\chi^\lambda(\pi)/d_\lambda| \leq \epsilon$ with high probability.\\
\end{minipage}

However, as we show in this section, problem 2 is efficiently solvable
using only classical randomized computation. Thus the techniques used
for problem 1 do not offer immediate benefit for problem 2. Although
this is in some sense a negative result, it provides an interesting
illustration of the difference in complexity between estimating
individual matrix elements of representations and estimating the
characters.

We can reduce problem 2 to problem 1 by sampling uniformly at random
from the standard Young tableaux compatible with Young diagram
$\lambda$. For each Young tableau sampled we estimate the
corresponding diagonal matrix element of $\rho_\lambda(\pi)$, as
described in problem 1. By averaging the diagonal matrix elements for
polynomially many samples, we obtain the normalized character to
polynomial precision. The problem of sampling uniformly at random from
the standard Young tableaux of a given shape is nontrivial but it has
been solved. Greene, Nijenhuis, and Wilf proved in 1979 that their
``hook-walk'' algorithm produces the standard Young tableaux of any
given shape with uniform probability\cite{Greene}. Examination of
\cite{Greene} shows that the time needed by the hook-walk algorithm to
produce a random standard Young tableaux compatible with a Young
diagram of $n$ boxes is upper bounded by $O(n^2)$.

By averaging over diagonal matrix elements we lose some information
contained in the individual matrix elements. This observation gives
the intuition that it should often be harder to estimate individual
matrix elements of a representation than to estimate its trace. Jones
polynomials provide an example in which this intuition is
confirmed. As discussed in \cite{Shor_Jordan}, computing 
the Jones polynomial of the trace closure of a braid reduces to
computing the normalized character of a certain representation of the braid
group. The problem of additively approximating this normalized
character is only DQC1-complete. In contrast, the individual matrix
elements of this representation yield the Jones polynomial of the plat
closure of the braid and are BQP-complete to approximate. We see a
very similar phenomenon in the symmetric group; problem 2 is  is
solvable by a randomized polynomial-time classical algorithm, whereas
problem 1 is not, as far as we know.

To construct a classical algorithm for problem 2, first recall that the
character of a given group element depends only on the element's
conjugacy class. We can think of any $\pi \in S_n$ as acting on the set
$\{1,2,\ldots,n\}$. The sizes of the orbits of the elements of
$\{1,2,\ldots,n\}$ under repeated application of $\pi$ form a
partition of the integer $n$. For example, consider the permutation
$\pi \in S_5$ defined by
\[
\begin{array}{lllll} 
\pi(1) = 2 & \pi(2) = 3 & \pi(3) = 1 & \pi(4) = 5 & \pi(5) = 4.
\end{array}
\]
This divides the set $\{1,2,3,4,5\}$ into the orbits $\{1,2,3\}$ and
$\{4,5\}$. Thus it corresponds to the partition $(3,2)$ of the integer
$5$. Two permutations in $S_n$ are conjugate if and only if they
correspond to the same partition. Thus, we can introduce the following
notation. For any two partitions $\mu$ and $\lambda$ of $n$ define 
$\chi_\mu^\lambda$ to be the irreducible character of $S_n$ corresponding
to the Young diagram of $\lambda$ evaluated at the conjugacy class
corresponding to $\mu$.

To obtain an efficient classical solution to problem 2 we use the
following theorem due to Roichman\cite{Roichman2}. 
%
\begin{theorem}[From \cite{Roichman2}]
\label{Roichman_rule}
For any partitions $\mu = (\mu_1,\ldots,\mu_l)$ and $\lambda =
(\lambda_1,\ldots,\lambda_k)$ of $n$, the corresponding irreducible
character of $S_n$ is given by
\[
\chi_\mu^\lambda = \sum_\Lambda W_\mu(\Lambda)
\]
where the sum is over all standard Young tableaux $\Lambda$ of shape
$\lambda$ and
\[
W_\mu(\Lambda) = \prod_{\stackrel{1 \leq i \leq k}{i \notin B(\mu)}} f_\mu(i,\Lambda)
\]
where 
$
B(\mu) = \{\mu_1 + \ldots + \mu_r | 1 \leq r \leq l \}
$ 
and
\[
f_\mu(i,\Lambda) = \left\{ \begin{array}{rl}
-1 & \textrm{box $i+1$ of $\Lambda$ is in the southwest of box $i$} \\
0 & \textrm{$i+1$ is in the northeast of $i$, $i+2$ is in the
  southwest of $i+1$, and $i+1 \notin B(\mu)$} \\
1 & \textrm{otherwise}
\end{array} \right.
\]
\end{theorem}
%
By using the hook walk algorithm we can sample uniformly at random
from the standard Young tableaux $\Lambda$ of shape $\lambda$. By
inspection of theorem \ref{Roichman_rule} we see that for each
$\Lambda$ sampled we can compute $W_\mu(\Lambda)$ classically in
$\mathrm{poly}(n)$ time. By averaging the values of $W_\mu(\Lambda)$
obtained during the course of the sampling we can thus obtain a
polynomially accurate additive approximation the the normalized
character, thereby solving problem 2.

Some readers may notice that theorem \ref{Roichman_rule} is similar in
form to the much older and better-known Murnaghan-Nakayama
rule. However, the Murnaghan-Nakayama rule is based on a sum over all
``rim-hook tableaux'' of shape $\lambda$ (see \cite{Roichman2}). It is
not obvious how to sample uniformly at random from the rim-hook
tableaux of a given shape. Thus, it is not obvious how to use the
Murnaghan-Nakayame rule to obtain a probabilistic classical algorithm
for problem 2.

\section{Lie Groups}

\subsection{Introduction}
\label{lieintro}

Because $U(n)$, $SU(n)$ and $SO(n)$ are compact linear groups, all of their
representations are unitary given the right choice of basis\cite{Artin9}.
In section \ref{Schurxform} we described how to efficiently approximate the
matrix elements from certain unitary irreducible representation of
$U(n)$. Here we present a more direct approach to this problem, which
can handle a larger set of representations of $U(n)$ and also extends
to some other compact Lie groups: $SU(n)$ and $SO(n)$.

$U(n)$, $SU(n)$, and $SO(n)$ are subgroups of $GL(n)$, the group of
all invertible $n \times n$ matrices. All of the irreducible
representations of $U(n)$ and $SU(n)$ can be obtained by restricting
the irreducible representations of $GL(n)$ to these subgroups. The 
best classical algorithms for computing irreducible representations of
$GL(n)$ and $U(n)$ appear to be those of \cite{Burgisser} and
\cite{Grabmeier}. These classical algorithms work by manipulating
matrices whose dimension equals the dimension of the
representation. Thus, they do not provide a polynomial time algorithm
for computing matrix elements from representations whose dimension is
exponentially large. The implementation of irreducible representations
of $SO(3)$ and $SU(2)$ by quantum circuits has been studied previously
by Zalka\cite{Zalka}.

\subsection{Gel'fand-Tsetlin representation of $U(n)$}
\label{Gelfand}

The irreducible representations of the Lie group $U(n)$
are most easily described in terms of the corresponding Lie algebra
$u(n)$. It is not necessary here delve into the theory of Lie
groups and Lie algebras, but those who are interested can see
\cite{Gilmore}. For now it suffices to say that $u(n)$ is the set of
all antihermitian $n \times n$ matrices, and for any $u \in U(n)$
there exists $h \in u(n)$ such that $u = e^h$. Given any representation
$a:u(n) \to u(m)$ one can construct a representation $A:U(n) \to U(m)$
as follows. For any $u \in U(n)$ find a corresponding $h(u) \in u(n)$
such that $e^h = u$,  and set $A(u) = e^{a(h(u))}$. If $a$ is an
antihermitian representation of $u(n)$ then $A$ is a unitary
representation of $U(n)$. Furthermore, it is clear that $A$ is
irreducible if and only if $a$ is irreducible.

It turns out that the irreducible representations of the algebra
$gl(n)$ of all $n \times n$ complex matrices remain irreducible when
restricted to the subalgebra $u(n)$. Furthermore, all of the
irreducible representations of $u(n)$ are obtained this way. Let
$E_{ij}$ be the $n \times n$ matrix with all matrix elements equal to
zero except for the matrix element in row $i$, column $j$, which is
equal to one. The set of all $n^2$ such matrices forms a basis over
$\mathbb{C}$ for $gl(n)$. Thus to describe a representation of $gl(n)$
it suffices to describe its action on each of the $E_{ij}$ matrices.

As described in chapter 18, volume 3 of \cite{Vilenkin}, explicit
matrix representations of $gl(n)$ were constructed by Gel'fand and
Tsetlin. (See also \cite{Gelfand_works}.) In their construction, one
thinks of the representation as 
acting on the formal span of a set of combinatorial objects called
Gel'fand patterns. The Gel'fand-Tsetlin representations of $E_{p,p-1}$
and $E_{p-1,p}$ are sparse and simple to compute for all $p \in
\{2,3,\ldots,n\}$. This property makes the Gel'fand-Tsetlin
representations particularly useful for quantum computation.

A Gel'fand pattern of width $n$ consists of $n$ rows of
integers\footnote{Some sources omit the top row, as it is left
  unchanged by the action of the representation.}. The $j\th$ row
(from bottom) has $j$ entries
$m_{1,j},m_{2,j},\ldots,m_{j,j}$. (Note that, in contrast to matrix
elements, the subscripts on the entries of Gel'fand patterns
conventionally indicate column first, then row.) These entries must
satisfy
\[
m_{j,n+1} \geq m_{j,n} \geq m_{j+1,n+1}.
\]
Gel'fand patterns are often written out diagrammatically. For example
the Gel'fand pattern of width 3 with rows
\[
\begin{array}{ccc}
m_{1,3} = 4 & m_{2,3} = 1 & m_{3,3} = 0 \\
m_{1,2} = 3 & m_{2,2} = 0 &             \\
m_{1,1} = 2 &             &             
\end{array}
\]
is represented by the diagram
\[
\left( \begin{array}{ccccc}
       4 &   & 1 &   & 0\\
         & 3 &   & 0 &  \\
         &   & 2 &   &
       \end{array} \right).
\]
This notation has the advantage that the entries that appear directly
to the upper left and upper right of a given entry form the upper and
lower bounds on the values that entry is allowed to take.
 
We call the top row of a Gel'fand pattern its weight\footnote{It is
  actually the \emph{highest} weight of the
  representation\cite{Vilenkin}, but for brevity I just call it the
  weight throughout this paper.}. To each weight of width $n$
corresponds one irreducible representation of $gl(n)$. This
irreducible representation acts on the formal span of all Gel'fand
patterns with that weight (of which there are always finitely
many). To describe the action of the representation of $gl(n)$ on
these patterns let
\begin{eqnarray}
\label{r1}
l_{p,q} & = & m_{p,q}-p \\
\label{r2}
a^j_{p-1} & = & \left| \frac{\prod_{i=1}^p (l_{i,p} - l_{j,p-1})
  \prod_{i=1}^{p-2} (l_{i,p-2}-l_{j,p-1}-1)} {\prod_{i \neq j}
  (l_{i,p-1} - l_{j,p-1}) (l_{i,p-1} - l_{j,p-1} -1)} \right|^{1/2} \\
\label{r3}
b^j_{p-1} & = & \left| \frac{\prod_{i=1}^p (l_{i,p} - l_{j,p-1} + 1)
  \prod_{i=1}^{p-2} (l_{i,p-2}-l_{j,p-1})} {\prod_{i \neq j}
  (l_{i,p-1} - l_{j,p-1}) (l_{i,p-1} - l_{j,p-1}+1) } \right|^{1/2}.
\end{eqnarray}
Let $M$ be a Gel'fand pattern and let $M_p^{+j}$ be the Gel'fand
pattern obtained from $M$ by replacing $m_{j,p}$ with
$m_{j,p}+1$. Similarly, let $M_p^{-j}$ be the Gel'fand pattern in
which $m_{j,p}$ has been replaced with $m_{j,p}-1$. The representation
$a_{\vec{m}}$ of $gl(n)$ corresponding to weight $\vec{m} \in
\mathbb{Z}^n$ is defined by the following rules\footnote{Warning:
  \cite{Vilenkin} contains a misprint, in which the sums in
  equations \ref{r4} and \ref{r5} are   taken up to $j=p$ instead of
  $j=p-1$.}, known as the Gel'fand-Tsetlin formulas.
\begin{eqnarray}
\label{r4}
a_{\vec{m}}(E_{p-1,p})M & = & \sum_{j=1}^{p-1} a^j_{p-1} M_{p-1}^{+j} \\
\label{r5}
a_{\vec{m}}(E_{p,p-1})M & = & \sum_{j=1}^{p-1} b^j_{p-1} M_{p-1}^{-j} \\
\label{r6}
a_{\vec{m}}(E_{p,p})M   & = & \left( \sum_{i=1}^p m_{i,p} -
\sum_{j=1}^{p-1} m_{j,p-1} \right) M
\end{eqnarray}

These formulas give implicitly a representation for all of $gl(n)$,
because any $E_{ij}$ can be obtained from operators of the form
$E_{p-1,p}$ and $E_{p,p-1}$ by using the commutation relation
$[E_{ik},E_{kl}] = E_{il}$. By restricting the representation
$a_{\vec{m}}$ to antihermitian subalgebra of $gl(n)$ and taking the
exponential, one obtains an irreducible group representation
$A_{\vec{m}}:U(n) \to U(d_{\vec{m}})$, where $d_{\vec{m}}$ is the
number of Gel'fand patterns with weight $\vec{m}$.

It should be noted that some references claim that
the set of allowed weights for representations of $GL(n)$ is
$\mathbb{N}^n$, whereas others identify, as we do, $\mathbb{Z}^n$ as
the allowed set of weights. The reason for this is that irreducible
representations of $GL(n)$ in which the entries $m_{n,1}, m_{n,2},
\ldots, m_{n,n}$ of the weight are all nonnegative are polynomial
invariants\cite{Konig37}. That is, for any $g \in GL(n)$ and any $\vec{m} \in
\mathrm{N}^n$, each matrix element of the representation
$\rho_{\vec{m}}(u)$ is a polynomial function of the $n^2$ matrix
elements of $u$. The representations involving negative weights are
called holomorphic representations, and many sources choose to neglect
them. In the case that $\vec{m} \in \mathbb{N}^n$, the Gel'fand
diagrams of width $n$ bijectively correspond to the semistandard Young
tableaux of $n$ rows (\emph{cf.} \cite{Difrancesco}, pg. 517).

\subsection{Quantum Algorithm for U(n)}
\label{ualg}

In this section we obtain an efficient quantum circuit implementation
of any irreducible representation of $U(n)$ in which the entries
$m_{1,n},\ldots,m_{n,n}$ of the highest weight are all at most
polynomially large. The dimension of such representations can grow
exponentially with $n$. Unlike the Schur transform, the method here
does not require $m_{1,n},\ldots,m_{n,n}$ to be nonnegative. We start
by finding a quantum circuit implementing the Gel'fand-Tsetlin
representation of an $n \times n$ unitary matrix of the form
\[
u_0 = \left[ \begin{array}{ccccc}
    u_{11} & u_{12} &   &        &   \\
    u_{21} & u_{22} &   &        &   \\
           &        & 1 &        &   \\
           &        &   & \ddots &   \\
           &        &   &        & 1
\end{array} \right],
\]
where all off-diagonal matrix elements not shown are zero. After that
we describe how to extend the construction to arbitrary 
$n \times n$ unitaries.

For a given weight $\vec{m} \in \mathbb{Z}^n$ we wish to implement the
corresponding representation $A_{\vec{m}}(u_0)$ with a quantum
circuit. To do this, we first find an $n \times n$ Hermitian matrix
$H_0$ such that $e^{i H_0} = u_0$. It is not hard to see that $H_0$
can be computed in polynomial time and takes the form
\[
H_0 = \left[ \begin{array}{cccccc}
             h_{11}   & h_{12} &        &       & \\
             h_{12}^* & h_{22} &        &       & \\
                     &       & 0      &       & \\
                     &       &        &\ddots & \\
                     &       &        &       & 0 
\end{array} \right].
\]
Thus,
\begin{equation}
\label{decomp1}
H_0 = h_{11} E_{11} + h_{12} E_{12} + h_{12}^* h_{21} + h_{22} E_{22}.
\end{equation}
Hence,
\begin{equation}
\label{decomp}
a_{\vec{m}}(H_0) = h_{11} a_{\vec{m}}(E_{11}) + h_{12}
a_{\vec{m}}(E_{12}) + h_{12}^* a_{\vec{m}}(E_{21}) + h_{22}
a_{\vec{m}}(E_{22}).
\end{equation}

To implement $A_{\vec{m}}(u_0)$ with a quantum circuit, we think of
$a_{\vec{m}}(H_0)$ as a Hamiltonian and simulate the corresponding
unitary time evolution $e^{-i a_{\vec{m}}(H_0)t}$ for $t=-1$. The
Hamiltonian $a_{\vec{m}}(H_0)$ has exponentially large dimension in
the cases of computational interest. However, examination of equation
\ref{decomp1} shows that $H_0$ is a linear combination of operators of
the form $E_{p,p-1}$ and $E_{p-1,p}$. Thus, by
the Gel'fand-Tsetlin rules of section \ref{Gelfand},
$a_{\vec{m}}(H_0)$ is sparse and that its individual matrix elements
are easy to compute. Under this circumstance, one can use the general
method for simulating sparse Hamiltonians proposed in
\cite{Aharonov_Tashma}.

Define row-sparse Hamiltonians to be those in which each row has at
most polynomially many nonzero entries. Further, define row-computable
Hamiltonians to be those such that there exists a polynomial time
algorithm which, given an index $i$, outputs a list of the nonzero
matrix elements in row $i$ and their locations. Clearly, all row
computable Hamiltonians are row-sparse. As shown in
\cite{Aharonov_Tashma}, the unitary $e^{-iHt}$ induced by any
row-computable Hamiltonian can be simulated in polynomial time
provided that the spectral norm $\|H\|$ and the time $t$ are at most
polynomially large. We have already noted that $a_{\vec{m}}(H_0)$ is
row-computable. $a_{\vec{m}}(H_0)$ is row sparse, and because we are
considering only polynomial highest weight, the entries of the
Gel'fand patterns, and hence the matrix elements of
$a_{\vec{m}}(H_0)$ are only polynomially large. Thus, by Gershgorin's
circle theorem $\|a_{\vec{m}}(H_0)\|$ is at most $\mathrm{poly}(n)$.

Having shown that a quantum circuit of $\mathrm{poly}(n)$ gates can
implement the Gel'fand-Tsetlin representation of an $n \times n$
unitary of the form $u_0$, the remaining task is to extend this to
arbitrary $n \times n$ unitaries. Examination of the preceding
construction shows that it works just the same for any unitary of the
form
\[
u_p = \id_p \oplus u \oplus \id_{n-p-2},
\]
where $\id_p$ denotes the $p \times p$ identity matrix and $u$ is a $2
\times 2$ unitary. Corresponding to $u_p$ is again an antihermitian
matrix of the form 
\[
H_p = 0_p \oplus h \oplus 0_{n-p-2}
\]
where $0_p$ is the $p \times p$ matrix of all zeros and $h$ is a $2
\times 2$ antihermitian matrix such that $e^{h} = u$. The only issue
to worry about is whether $\|a_{\vec{m}}(H_p)\|$ is at most
$\mathrm{poly(n)}$. By symmetry, one expects that
$\|a_{\vec{m}}(H_p)\|$ should be independent of $p$. However, this is
not obvious from examination of equations \ref{r1} through
\ref{r6}. Nevertheless, it is true, as shown in appendix
\ref{normindep}. Thus, the norm is no different than in the $p=0$
case, \emph{i.e.} $H_0$. 

By concatenating the quantum circuits implementing
$A_{\vec{m}}(u_1), A_{\vec{m}}(u_2),\ldots, A_{\vec{m}}(u_L)$, one can
implement $A_{\vec{m}}(u_1 u_2 \ldots u_L)$. We next show that any $n
\times n$ unitary can be obtained as a product of $\mathrm{poly}(n)$
matrices, each of the form $u_p$, thus showing that the quantum
algorithm is completely general and always runs in polynomial time. 

For any $2 \times 2$ matrix $M$, let $\mathcal{E}(M,i,j)$ be the $n
\times n$ matrix in which $M$ acts on the $i\th$ and $j\th$ basis
vectors. In other words, the $k,l$ matrix element of
$\mathcal{E}(M,i,j)$ is
\[
\mathcal{E}(M,i,j)_{kl} = \left\{ \begin{array}{ll} 
M_{11} & \textrm{if $k=i$ and $l=i$} \\
M_{12} & \textrm{if $k=i$ and $l=j$} \\
M_{21} & \textrm{if $k=j$ and $l=i$} \\
M_{22} & \textrm{if $k=j$ and $l=j$} \\
\delta_{kl} & \textrm{otherwise}
\end{array} \right. .
\]
Thus
\[
u_p = \mathcal{E} \left( \left[ \begin{array}{cccc}
u_{11} & u_{12} \\
u_{21} & u_{22} \end{array} \right], m+1,m+2 \right).
\]
Next note that,
\[
\begin{array}{l}
\mathcal{E} \left( \left[ \begin{array}{cccc}
u_{11} & u_{12} \\
u_{21} & u_{22} \end{array} \right],m+1,m+3 \right) = \vspace{3pt} \\
\mathcal{E} \left( \left[ \begin{array}{cccc}
0 & 1 \\
1 & 0 \end{array} \right],m+2,m+3 \right)
\mathcal{E} \left( \left[ \begin{array}{cccc}
u_{11} & u_{12} \\
u_{21} & u_{22} \end{array} \right], m+1,m+2 \right)
\mathcal{E} \left( \left[ \begin{array}{cccc}
0 & 1 \\
1 & 0 \end{array} \right],m+2,m+3 \right).
\end{array}
\]
Thus the matrix 
\[
\mathcal{E} \left( \left[ \begin{array}{cccc}
u_{11} & u_{12} \\
u_{21} & u_{22} \end{array} \right],m+1,m+3 \right)
\]
is obtained as a product of three matrices of the form
$u_p$. By repeating this conjugation process, one can obtain
\begin{equation}
\label{twolevel}
\mathcal{E} \left( \left[ \begin{array}{cccc}
u_{11} & u_{12} \\
u_{21} & u_{22} \end{array} \right],i,j \right)
\end{equation}
for arbitrary $i,j$ as a product of one matrix of the form
\[
\mathcal{E} \left( \left[ \begin{array}{cccc}
u_{11} & u_{12} \\
u_{21} & u_{22} \end{array} \right],p+1,p+2 \right)
\]
for some $p$ and at most $O(n)$ matrices of the form
\[
\mathcal{E} \left( \left[ \begin{array}{cccc}
0 & 1 \\
1 & 0 \end{array} \right],q+1,q+2 \right)
\]
with various $q$. A matrix of the form shown in equation
\ref{twolevel} is called a two-level unitary. As shown in section
4.5.1 of \cite{Nielsen_Chuang}, any $n \times n$ unitary is obtainable
as a product of $\mathrm{poly}(n)$ two-level unitaries. Thus we obtain
$A_{\vec{m}}(U)$ for any $n \times n$ unitary $U$ using
$\mathrm{poly}(n)$ quantum gates. One can then obtain any matrix
element of $A_{\vec{m}}(U)$ to precision $\pm \epsilon$ by repeating
the Hadamard test $O(1/\epsilon^2)$ times. 

\subsection{Special Orthogonal Group}

The special orthogonal group $SO(n)$ consists of all $n \times n$
real orthogonal matrices with determinant equal to one. The
irreducible representations of $SO(n)$ are closely related to those of
$U(n)$ and can also be expressed unitarily using a Gel'fand-Tsetlin
basis. As discussed in chapter 18, volume 3 of \cite{Vilenkin}, the
nature of the representations of $SO(n)$ depends on whether $n$ is even
or odd. Following \cite{Vilenkin} and \cite{Gelfand_works}, we
therefore introduce an integer $k$ and consider $SO(2k+1)$ and
$SO(2k)$ separately.

The irreducible representations of $SO(2k+1)$ are in bijective
correspondence with the set of allowed weight vectors $\vec{m}$
consisting of $k$ entries, each of which is an integer or
half-integer. Furthermore, the entries must satisfy
\[
m_{1,n} \geq m_{2,n} \geq \ldots \geq m_{k,n} \geq 0.
\]
The irreducible representations of $SO(2k)$ correspond to the weight
vectors $\vec{m}$ with $k-1$ entries, each of which must be an integer
or half integer, and which must satisfy
\[
m_{1,n} \geq m_{2,n} \geq \ldots \geq m_{k-1,n} \geq |m_{k,n}|.
\]

As in the case of $U(n)$, the set of allowed Gel'fand patterns is
determined by rules for how a row can compare to the one above it. For
$SO(n)$ these rules are slightly more complicated, and the rule for
the $j\th$ row depends on whether $j$ is odd or even. Specifically the
even rule for $j=2k$ is
\[
m_{1,2k+1} \geq m_{1,2k} \geq m_{2,2k+1} \geq m_{2,3k} \geq \ldots
\geq m_{k,2k+1} \geq m_{k,2k} \geq -m_{k,2k-1},
\]
and the odd rule for $j=2k-1$ is
\[
m_{1,2k} \geq m_{1,2k-1} \geq m_{2,2k} \geq m_{2,2k-1} \geq \ldots
\geq m_{k-1,2k} \geq m_{k-1,2k-1} \geq |m_{k,2k}|.
\]
The Lie algebra $so(n)$ corresponding to the Lie group $SO(n)$ is the
algebra of all antisymmetric $n \times n$ matrices. For any
$G \in SO(n)$ there exists a $g \in so(n)$ such that $e^{g} =
G$. The Lie algebra $so(n)$ is the space of all $n \times n$ real
traceless antisymmetric matrices. Thus it is spanned by operators of
the form
\[
I_{k,i} = E_{i,k}-E_{k,i} \quad 1 \leq i < k \leq n.
\]
We can fully specify a representation of $so(n)$ by specifying the
representations of the operators of the form $I_{q+1,q}$ because these
generate $so(n)$. That is, any element of $so(n)$ can be obtained as a
linear combination of commutators of such operators. The
Gel'fand-Tsetlin representation $b_{\vec{m}}$ of these operators
depends on whether $q$ is even or odd, and is given by the following
formulas.
\begin{eqnarray*}
A_{2p}^j(M) & = & \frac{1}{2} \left| \frac{\prod_{r=1}^{p-1} 
\left[(l_{r,2p-1}-\frac{1}{2})^2-(l_{j,2p}+\frac{1}{2})^2\right]
\prod_{r=1}^p \left[(l_{r,2p+1}-\frac{1}{2})^2 -
  (l_{j,2p}+\frac{1}{2})^2 \right]}
{\prod_{r \neq j} (l_{r,2p}^2 - l_{j,2p}^2)
  (l_{r,2p}^2-(l_{j,2p}+1)^2)} \right|^{1/2} \\
B_{2p+1}^j(M) & = & \left| \frac{\prod_{r=1}^p (l_{r,2p}^2-l_{j,2p+1}^2)
  \prod_{r=1}^{p+1} (l_{r,2p+2}^2-l_{j,2p+1}^2)}
{l_{j,2p+1}^2 (4l_{j,2p+1}^2-1) \prod_{r \neq j}
  (l_{r,2p+1}^2-l_{j,2p+1}^2) (l_{j,2p+1}^2 - (l_{r,2p+1}-1)^2)}
\right|^{1/2} \\
C_{2p}(M) & = & \frac{\prod_{r=1}^p l_{r,2p} \prod_{r=1}^{p+1}
  l_{r,2p+2}}{\prod_{r=1}^p l_{r,2p+1} (l_{r,2p+1}-1)} \\
b_{\vec{m}}(I_{2p+1,2p}) M & = & \sum_{j=1}^p A_{2p}^j(M) M_{2p}^{+j} -
\sum_{j=1}^p A_{2p}^j(M_2p^{-j}) M_{2p}^{-j} \\
b_{\vec{m}}(I_{2p+2,2p+1}) M & = & \sum_{j=1}^p B_{2p+1}^j(M)
M_{2p+1}^{+j} - \sum_{j=1}^p B_{2p+1}^j(M_{2p+1}^{-j}) M_{2p+1}^{-j} +
i C_{2p}(M) M
\end{eqnarray*}

By applying these rules to the set of allowed Gel'fand patterns
described above one obtains the irreducible representations of the
algebra $so(n)$. By exponentiating these, one then obtains the
irreducible representations of the group $SO(n)$. Thus the quantum
algorithm for approximating the matrix elements of the irreducible
representations of $SO(n)$ is analogous to that for $U(n)$.

\subsection{Special Unitary Group}
\label{sun}

The irreducible representations of $SU(n)$ can be easily constructed
from the irreducible representations of $U(n)$, using the following
facts taken from chapter 10 of \cite{BR}. The representations of
$U(n)$ can be partitioned into a set of equivalence classes of
projectively equivalent representations. Two representations of $U(n)$
with weights $\vec{l} =(l_1,l_2,\ldots,l_n)$ and
$\vec{m}=(m_1,m_2,\ldots,m_n)$ are projectively equivalent if 
and only if there exists some integer $s$ such that $m_i = l_i + s$
for all $1 \leq i \leq n$. Any irreducible representation of $U(n)$ remains 
irreducible when restricted to $SU(n)$. Furthermore, by choosing one
representative from each class of projectively equivalent
representations of $U(n)$ and restricting to $SU(n)$ one obtains a
complete set of inequivalent irreducible representations of
$SU(n)$. The Lie algebra $su(n)$ corresponding to the Lie group
$SU(n)$ is easily characterized; it is the space of all traceless $n
\times n$ antihermitian matrices. Thus the matrix elements of the
irreducible representations of $SU(n)$ are obtained by essentially the
same quantum algorithm given for $U(n)$ in section \ref{ualg}. 

\subsection{Characters of Lie Groups}

As always, an algorithm for approximating matrix elements immediately
gives us an algorithm for approximating the normalized
characters. However, the characters of $U(n)$, $SU(n)$, and $SO(n)$ are
classically computable in $\mathrm{poly}(n)$ time. As discussed in
\cite{Fulton_Harris}, the characters of any compact Lie group are given
by the Weyl character formula. In general this formula may involve
sums of exponentially many terms. However, in the special cases of
$U(n)$, $SU(n)$, and $SO(n)$ the formula reduces to simpler
forms\cite{Fulton_Harris}, given below.

Because characters depend only on conjugacy class, the character
$\chi_{\vec{m}}(u)$ depends only on the eigenvalues of $u$. For $u \in
U(n)$ let $\lambda_1,\ldots,\lambda_n$ denote the eigenvalues. Let
$\vec{m} = (m_1,m_2,\ldots,m_n) \in \mathbb{Z}^n$ be the weight of a
representation of $U(n)$. Let 
\begin{equation}
\label{l}
l_i = m_i + n - i
\end{equation}
for each $i \in \{1,2,\ldots,n\}$. The character of the representation
of weight $\vec{m}$ is
\[
\chi^{U(n)}_{\vec{m}}(u) = \frac{\det A}{\det B}
\]
where $A$ and $B$ are the following $n \times n$
matrices
\begin{eqnarray*}
A_{ij} & = & \lambda_i^{l_j} \\
B_{ij} & = & \lambda_i^{n-j}.
\end{eqnarray*}

This formula breaks down if $u$ has a degenerate spectrum. However,
the value of the character for degenerate $u$ can be obtained by
taking the limit as some eigenvalues converge to the same value. As
shown in \cite{Weyl7}, one can obtain the dimension $d_{\vec{m}}$ of the
representation corresponding to a given weight $\vec{m}$ by
calculating $\lim_{u \to \id} \chi_{\vec{m}}(u)$. Specifically, by
choosing $\lambda_j = e^{ij \epsilon}$ for each $1 \leq j \leq n$ and
taking the limit as $\epsilon \to 0$ one obtains
\[
d_{\vec{m}} = \frac{\prod_{i<j} (l_j - l_i)}
{\prod_{i<j} (j-i)},
\]
where $l_i$ is as defined in equation \ref{l}.

As discussed in section \ref{sun}, the irreducible representations of
$SU(n)$ are restrictions of irreducible representations of $U(n)$,
therefore the characters of $SU(n)$ are given by the same formula as
the characters of $U(n)$.

$SO(n)$ consists of real matrices. The characteristic polynomials of
these matrices have real coefficients, and thus their roots come in
complex conjugate pairs. Thus, the eigenvalues of an element $g \in
SO(2k+1)$ take the form
\[
\lambda_1, \lambda_2, \ldots, \lambda_k, 1, \lambda_1^*, \lambda_2^*,
\ldots, \lambda_k^*,
\]
and for $g \in SO(2k)$, the eigenvalues take the form
\[
\lambda_1, \lambda_2, \ldots, \lambda_k, \lambda_1^*, \lambda_2^*,
\ldots, \lambda_k^*.
\]
As discussed in \cite{Fulton_Harris}, the characters of the special
orthogonal group are given by
\[
\chi_{\vec{m}}^{SO(2k+1)}(g) = \frac{\det C}{\det D}
\]
and
\[
\chi_{\vec{m}}^{SO(2k)}(g) = \frac{\det E + \det F}
{\det G}
\]
where $C$ and $D$ are the following $k \times k$ matrices
\begin{eqnarray*}
C_{ij} & = & \lambda_j^{m_i+n-i+1/2} -
\lambda_j^{-(m_i+n-i+1/2)} \\
D_{ij} & = & \lambda_j^{n-i+1/2} - \lambda_j^{-(n-i+1/2)}
\end{eqnarray*}
and $E,F,G$ are the following $(k-1) \times (k-1)$ matrices
\begin{eqnarray*}
E_{ij} & = & \lambda_j^{l_i} + \lambda_j^{-l_i}\\
F_{ij} & = & \lambda_j^{l_i} - \lambda_j^{-l_i}\\
G_{ij} & = & \lambda_j^{n-i} + \lambda_j^{-(n-i)},
\end{eqnarray*}
where $l_i$ is as defined in equation \ref{l}. 

As with $U(n)$, the character of any element with a degenerate
spectrum can be obtained by taking an appropriate limit.

\subsection{Open Problems Regarding Lie groups}

The quantum circuits presented in the preceeding sections efficiently
implement the irreducible representations of $U(n)$, $SU(n)$, and
$SO(n)$ that have polynomial highest weight and polynomial $n$. It is
an interesting open problem to implement irreducible representations
with quantum circuits that scale polynomially in the number of digits
used to specify the highest weight. Alternatively, one could try to
implement an Schur transform to handle exponential highest weight,
which is also an open problem. It is even concievable that Schur-like
transforms could be efficiently implemented for exponential $n$. That
is, there could exist a quantum circuit of $\mathrm{polylog}(n)$
gates implementing a unitary transform $V$ such that for any $U \in
U(n)$, $V U V^{-1}$ is a direct sum of irreducible representations of
$U$. Of course, if $n$ is exponentially large, than we cannot have an
explicit description of $U$, rather the group element $U$ could itself
be defined by a quantum circuit.

A completely different open problem is presented by the symplectic
group. Having constructed quantum circuits for $SO(n)$ and $SU(n)$,
the symplectic group is the only ``classical'' Lie group remaining to
be analyzed. Thus it is natural to ask whether its irreducible
representations can be efficiently implemented by quantum
circuits. Two different groups can go by the name symplectic group
depending on the reference. Connected non-compact simple Lie groups
have no nontrivial finite-dimensional unitary representations (see
\cite{BR}, theorem 8.1.2). This applies to one of the groups that goes
by the name of symplectic. On the other hand, the irreducible
representations of the compact symplectic group seem promising for
implementation by quantum circuits. The main task seems to be finding
a basis for these representations that is subgroup adapted and makes
the representations unitary. A non-unitary subgroup-adapted basis is
given in \cite{Molev}.

\section{Alternating Group}
\label{althyp}

In section \ref{fourier}, we described a method to approximate matrix
elements of the irreducible representations of the symmetric group
using the symmetric group quantum Fourier transform. Here we take a
more direct approach to this problem, which extends to the
alternating group. To do this we must first
explicitly describe the Young-Yamanouchi representation of the
symmetric group.

\subsection{Young-Yamanouchi Representation}
\label{Young}

For a given Young diagram $\lambda$, let $\mathcal{V}_\lambda$ be the
vector space formally spanned by all standard Young tableaux
compatible with $\lambda$. For example, if
\[
\lambda = \begin{array}{l} \includegraphics[width=0.3in]{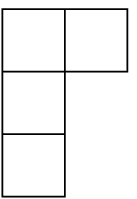}
\end{array}
\]
then $\mathcal{V}_\lambda$ is the 3-dimensional space consisting of
all formal linear combinations of
\[
\begin{array}{l} \includegraphics[width=1.4in]{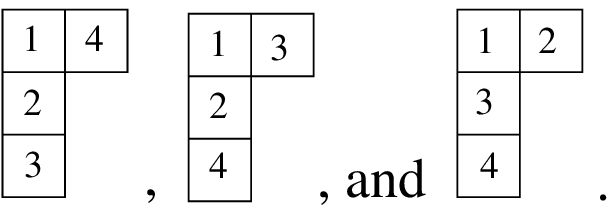}
\end{array}
\]
For any given Young diagram $\lambda$, the corresponding irreducible
representation in the Young-Yamanouchi basis is a homomorphism
$\rho_\lambda$ from $S_n$ to the group of orthogonal linear
transformations on $\mathcal{V}_\lambda$. It is not easy to directly
compute $\rho_\lambda(\pi)$ for an arbitrary permutation
$\pi$. However, it is much easier to compute the representation of a 
transposition of neighbors. That is, we imagine the elements of $S_n$
as permuting a set of objects $1,2,\ldots,n$, arranged on a line. A
neighbor transposition $\sigma_i$ swaps objects $i$ and $i+1$. It is
well known that the set $\{ \sigma_1,\sigma_2,\ldots,\sigma_{n-1} \}$
generates $S_n$.

The matrix elements for the Young-Yamanouchi representation of
transpositions of neighbors can be obtained using a single simple
rule: Let $\Lambda$ be any standard Young tableau compatible with Young diagram
$\lambda$ then
\begin{equation}
\label{rule}
\rho_\lambda(\sigma_i) \Lambda = \frac{1}{\tau_i^\Lambda} \Lambda +
\sqrt{1-\frac{1}{(\tau_i^\Lambda)^2}} \Lambda',
\end{equation}
where $\Lambda'$ is the Young tableau obtained from $\Lambda$ by
swapping boxes $i$ and $i+1$, and $\tau_i^\Lambda$ is the axial
distance from box $i+1$ to box $i$. That is, we
are allowed to hop vertically or horizontally to nearest
neighbors, and $\tau$ is the number of hops needed to get from box $i+1$
to box $i$, where going down or left counts as $+1$ hop and going
up or right counts as $-1$ hop. To illustrate the use of equation
\ref{rule}, some examples are given in figure \ref{examples}.

\begin{figure}
\begin{center}
\includegraphics[width=0.85\textwidth]{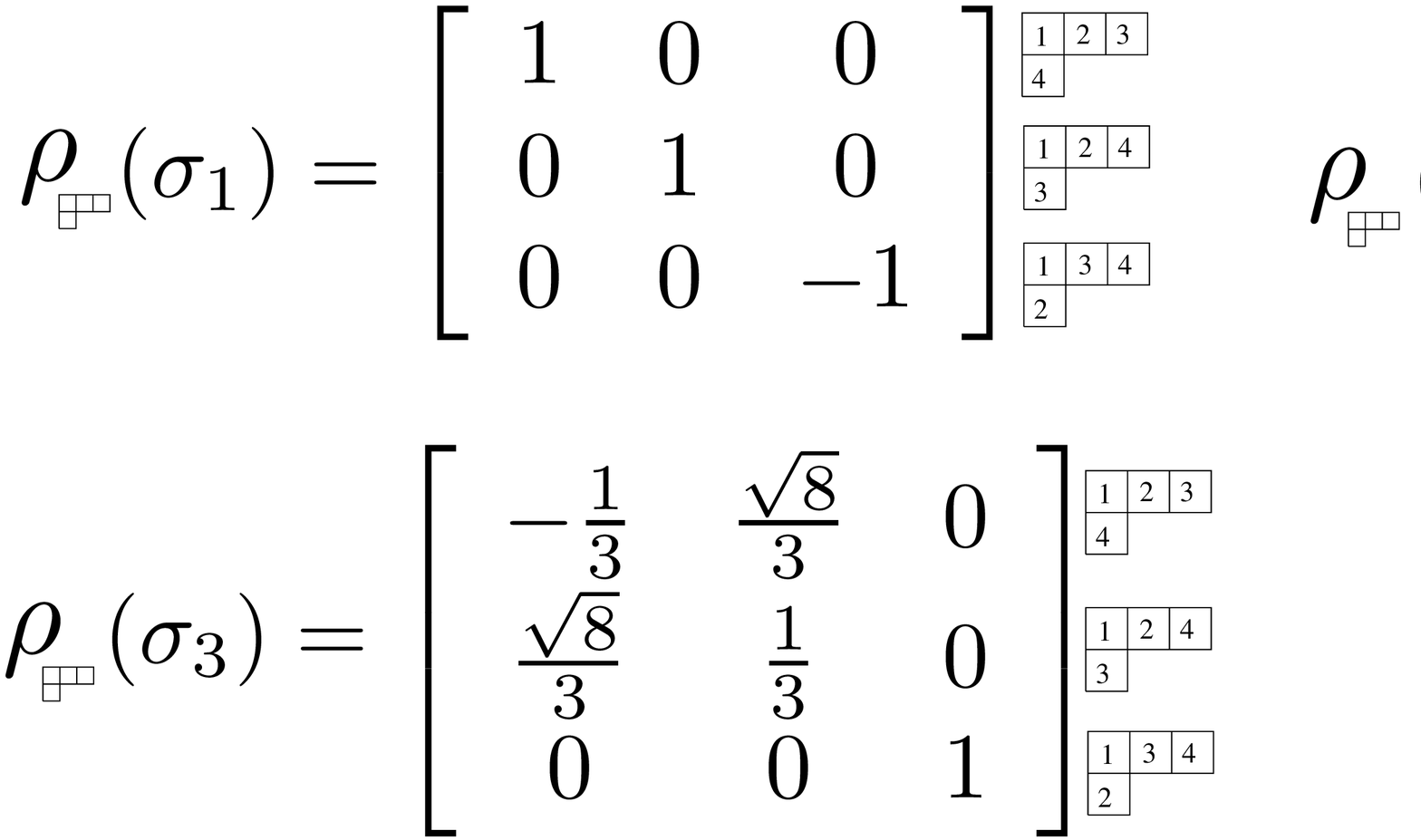}
\caption{\label{examples} 
The above matrices are irreducible representations in the
Young-Yamanouchi basis with Young diagram \protect
\includegraphics[width=0.18in]{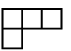}. Here $\sigma_i$ 
  is the permutation in $S_4$ that swaps $i$ with $i+1$.}
\end{center}
\end{figure}

In certain cases, starting with a standard Young tableau and swapping
boxes $i$ and $i+1$ does not yield a standard Young tableau, as
illustrated below.
\[
\includegraphics[width=2.7in]{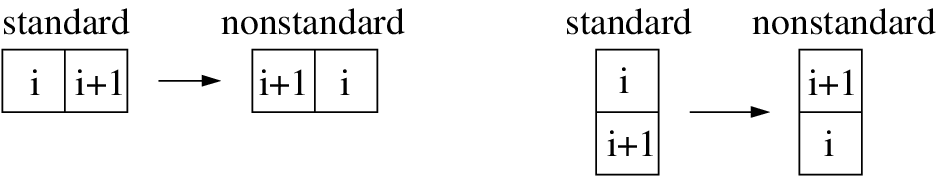}
\]
Some thought shows that all such cases are of one of the two types
shown above. In both of these types, the axial distance is $\pm
1$. By equation \ref{rule}, the coefficient on the invalid Young
tableau is $\sqrt{1-\frac{1}{(\pm 1)^2}} = 0$. Thus the representation
lies strictly within the space of standard Young tableaux.

\subsection{Direct Quantum Algorithm for $S_n$}
\label{algorithm}

We can directly implement the irreducible representations of $S_n$ by
first decomposing the given permutation into a product of
transposition of neighbors. The classical bubblesort algorithm
achieves this efficiently. For any permutation in $S_n$, it yields a
decomposition consisting of at most $O(n^2)$ transpositions. As seen
in the previous section, the Young-Yamanouchi representation of any
transposition is a direct sum of $2 \times 2$ and $1 \times 1$ blocks,
and the matrix elements of these blocks are easy to compute. As shown
in \cite{Aharonov_Tashma}, any unitary with these properties may be
implemented by a quantum circuit with polynomially many gates. By
concatenating at most $O(n^2)$ such quantum circuits we obtain the
representation of any permutation in $S_n$. The Hadamard test allows a
measurement to polynomial precision of the matrix elements of this
representation.

\subsection{Algorithm for Alternating Group}
\label{alt}

Any permutation $\pi$ corresponds to a permutation matrix with matrix
element $i,j$ given by $\delta_{\pi(i),j}$. The determinant of any
permutation matrix is $\pm 1$, and is known as the sign of the
permutation. The permutations of sign $+1$ are called even, and the
permutations of sign $-1$ are called odd. This is because a
transposition has determinant $-1$, and therefore any product of an
odd number of transpositions is odd and any product of an even number
of transpositions is even.

The even permutations in $S_n$ form a subgroup called the alternating
group $A_n$, which has size $n!/2$. $A_n$ is a simple group
(\emph{i.e.} it contains no normal subgroup) and it is the only normal
subgroup of $S_n$ other than $\{ \id \}$ and $S_n$. As one might
guess, the irreducible representations of the alternating group are
closely related to the irreducible representations of the symmetric
group. Consequently, as shown in this section, the quantum algorithm
of section \ref{algorithm} can be easily adapted to approximate any
matrix element of any irreducible representation of $A_n$ to within
$\pm \epsilon$ in $\mathrm{poly}(n,1/\epsilon)$ time.

Explicit orthogonal matrix representations of the alternating group
are worked out in \cite{Thrall} and recounted nicely in
\cite{Headley}. Any representation $\rho$ of $S_n$ is automatically
also a representation of $A_n$. However an irreducible representation
$\rho$ of $S_n$ may no longer be irreducible when restricted to
$A_n$. Each irreducible representation of $S_n$ either remains
irreducible when restricted to $A_n$ or decomposes into a direct sum
of two irreducible representations of $A_n$. All of the irreducible
representations of $A_n$ are obtained in this way. 

\begin{figure}
\begin{center}
\includegraphics[width=0.3\textwidth]{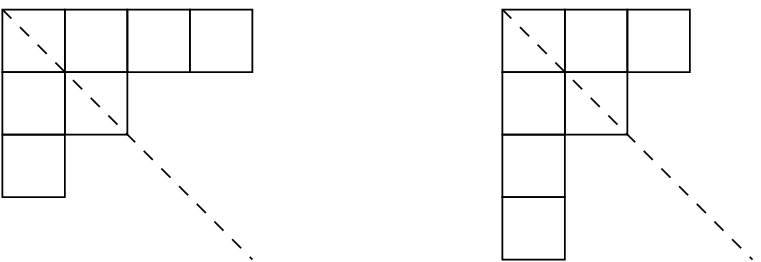}
\caption{\label{conjugate} To obtain the conjugate $\hat{\lambda}$ of Young
  diagram $\lambda$, reflect $\lambda$ about its diagonal. In other
  words the number of boxes in the $i\th$ column of $\hat{\lambda}$ is
  equal to the number of boxes in the $i\th$ row of $\lambda$.}
\end{center}
\end{figure}

The conjugate of Young diagram $\lambda$ is obtained by reflecting
$\lambda$ about the main diagonal, as shown in figure \ref{conjugate}. If
$\lambda$ is not self-conjugate then the representation
$\rho_{\lambda}$ of $S_n$ remains irreducible when restricted to
$A_n$. In this case we can simply use the algorithm of section
\ref{algorithm}. If $\lambda$ is self-conjugate then the
representation $\rho_\lambda$ of $S_n$ becomes reducible when
restricted to $A_n$. It is a direct sum of two irreducible
representations of $A_n$, called $\rho_{\lambda+}$ and
$\rho_{\lambda-}$. The two corresponding invariant subspaces of the
reducible representation are the $+1$ and $-1$ eigenspaces,
respectively, of the ``associator'' operator $S$ defined as follows.

\begin{figure}
\begin{center}
\includegraphics[width=0.12\textwidth]{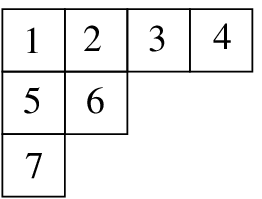}
\caption{\label{typewriter} For a given Young diagram, there is a
  unique Young tableau in ``typewriter'' order, in which the boxes are
  numbered from left to right across the top row then from left to
  right across the next row, and so on, as illustrated in the example
  above.}
\end{center}
\end{figure}

Let $\lambda$ be a self-conjugate Young diagram of $n$ boxes. Let
$\Lambda_0$ be the ``typewriter-order'' Young tableau obtained by numbering
the boxes from left to right across the first row, then left to right
across the second row, and so on, as illustrated in figure
\ref{typewriter}. For any standard Young tableau $\Lambda$ of shape
$\lambda$, let $w_\Lambda \in S_n$ be the permutation that brings the
boxes into typewriter order. That is, $w_{\Lambda} \Lambda =
\Lambda_0$. Let $\hat{\Lambda}$ be the conjugate of $\Lambda$, obtained by
reflecting $\Lambda$ about the main diagonal. If $\Lambda$ is standard
then so is $\hat{\Lambda}$. Let $d(\lambda)$ be the length of the main
diagonal of $\lambda$. $S$ is the linear operator on
$\mathcal{V}_\lambda$ defined by
\begin{equation}
\label{S}
S \Lambda = i^{(n-d(\lambda))/2} \mathrm{sign}(w_\Lambda) \hat{\Lambda}.
\end{equation}

An orthonormal basis for each of the eigenspaces of $S$ can be easily
constructed from the Young-Yamanouchi basis. When $(n-d(\lambda))/2$
is odd, every standard Young tableau $\Lambda$ of shape $\lambda$ has
the property $\mathrm{sign}(w_\Lambda) =
-\mathrm{sign}(w_{\hat{\Lambda}})$, and $S$ is a direct sum of $2
\times 2$ blocks of the form
\[
\left[ \begin{array}{cc} 0 & -i \\
                         i & 0 
\end{array} \right]
\]
interchanging $\Lambda$ and $\hat{\Lambda}$. In this case, the linear
combinations 
$
\frac{1}{\sqrt{2}} (\Lambda + i \hat{\Lambda})
$
for each conjugate pair of standard Young tableaux form an orthonormal
basis for the $+1$ eigenspace of $S$, and the linear combinations
$
\frac{1}{\sqrt{2}}(\Lambda - i \hat{\Lambda})
$
form an orthonormal basis for the $-1$ eigenspace of $S$. Similarly,
when $(n-d(\lambda))/2$ is even, $\mathrm{sign}(w_\Lambda) =
\mathrm{sign}(w_{\hat{\Lambda}})$ for all standard Young tableaux
$\Lambda$ of shape $\lambda$. Thus $S$ is a direct sum of $2 \times 2$
blocks of the form
\[
\left[ \begin{array}{cc} 0 & -1 \\
                        -1 & 0
\end{array} \right]
\]
interchanging $\Lambda$ and $\hat{\Lambda}$. In this case the linear
combinations 
$
\frac{1}{\sqrt{2}} (\Lambda - \hat{\Lambda})
$
form an orthonormal basis for the $+1$ eigenspace of $S$ and the
linear combinations
$
\frac{1}{\sqrt{2}} (\Lambda + \hat{\Lambda})
$
form an orthonormal basis for the $-1$ eigenspace of $S$.

Suppose $\lambda$ is self-conjugate and $(n-d(\lambda))/2$ is
even. Any matrix element of the irreducible representation
$\rho_{\lambda+}$ of $A_n$ is given by
\[
\frac{1}{2} (\Lambda + \hat{\Lambda}) \rho_\lambda(\pi) (\Gamma + \hat{\Gamma}),
\]
where $\Lambda, \Gamma$ is some pair of standard Young tableaux and
$\pi$ is some element of $A_n$. This is a linear combination of only
four Young-Yamanouchi matrix elements of $\rho_\lambda(\pi)$. One
can use the algorithm of section \ref{algorithm} to calculate each of
these and then simply add them up with the appropriate coefficients.
The cases where $(n-d(\lambda))/2$ is odd and/or we want a matrix
element of $\rho_{\lambda-}$ are analogous.

\section{Acknowledgements}

I thank Greg Kuperberg and anonymous referees for suggesting the
approaches described in sections \ref{fourier} and \ref{Schurxform}. I
thank Daniel Rockmore, Cris Moore, Andrew Childs, Aram Harrow, John
Preskill, and Jeffrey Goldstone for useful discussions. I thank Isaac
Chuang and Vincent Crespi for comments that helped to inspire this
work, and anonymous referees for useful comments. Parts of this work
were completed the Center for Theoretical physics at MIT, the Digital
Materials Laboratory at RIKEN, and the Institute for Quantum
Information at Caltech. I thank these institutions as well as the Army
Research Office (ARO), the Disruptive Technology Office (DTO), the
Department of Energy (DOE), and Franco Nori and Sahel Ashab at RIKEN.

\appendix

\section{$\|a_{\vec{m}}(H_p)\|$ is independent of $p$}
\label{normindep}

As shown in section \ref{ualg}, the irreducible representation of
an arbitrary $u \in U(n)$ with weight $\vec{m}$ can be computed by
simulating the time evolution according to a series of Hamiltonians of
the form $A_{\vec{m}}(H_p)$, where $A_{\vec{m}}$ is the
Gel'fand-Tsetlin representation of the Lie algebra $su(n)$ and
\[
H_p = 0_p \oplus h \oplus 0_{n-p-2},
\]
where $h$ is a $2 \times 2$ antihermitian matrix. The quantum
algorithm for simulating these Hamiltonians require that
$\|A_{\vec{m}}(H_p)\|$ be at most $\mathrm{poly}(n)$. In section
\ref{ualg} we showed this to be the case for $p=0$. Here we prove
it for all $p$ by showing:

\begin{proposition}
Let $h$ be a fixed $2 \times 2$ antihermitian matrix and let $H_p =
0_p \oplus h \oplus 0_{n-p-2}$. Let $a_{\vec{m}}$ be the
Gel'fand-Tsetlin representation of $su(n)$ with weight $\vec{m}$. Then
$\|a_{\vec{m}}(H_p)\|$ is independent of $p$.
\end{proposition}

\begin{proof}

Let $U_p^k = e^{k H_p}$. Then
\[
U_p^k = \id_p \oplus e^{k h} \oplus \id_{n-p-2}.
\]
Thus for any $0 \leq q \leq n$, there exists $V \in U(n)$ such that
\begin{equation}
\label{permutebasis}
U_q^k = V U_p^k V^{-1}.
\end{equation}
Specifically, $V$ is just a permutation matrix. Let $A_{\vec{m}}$ be
the Gel'fand-Tsetlin representation of $SU(n)$. That is,
\[
A_{\vec{m}}(U_q^k) = e^{a_{\vec{m}}(k H_q)}.
\]
Thus
\begin{eqnarray}
\left\| \frac{\ud}{\ud k} A_{\vec{m}}(U_p^k) \right\| & = & 
\| a_{\vec{m}}(H_p) e^{k a_{\vec{m}}(H_p)} \| \nonumber \\
\label{normderiv}
& = & \|a_{\vec{m}}(H_p)\|.
\end{eqnarray}
Here we have used the fact that $A_{\vec{m}}$ is a unitary
representation. Similarly,
\[
\|a_{\vec{m}}(H_q)\| = \left\| \frac{\ud}{\ud k} A_{\vec{m}}(U_q^k) \right\|.
\]
Using equation \ref{permutebasis}, this is equal to
\[
\left\| \frac{\ud}{\ud k} A_{\vec{m}}(V U_p^k V^{-1}) \right\|.
\]
Because $A_{\vec{m}}$ is a group homomorphism and $V$ is
independent of $k$ this is equal to
\[
\left\| A_{\vec{m}}(V) \left( \frac{\ud}{\ud k} A_{\vec{m}}(U_p^k) \right)
  A_{\vec{m}}(V)^{-1} \right\|.
\]
Because $A_{\vec{m}}$ is a unitary representation this is equal to
\[
\left\| \frac{\ud}{\ud k} A_{\vec{m}}(U_p^k) \right\|.
\]
By equation \ref{normderiv} this is equal to 
$\| a_{\vec{m}}(H_p) \|$. \end{proof}

\bibliography{irreps}

\end{document}